\documentclass[10pt,conference]{IEEEtran}
\usepackage{cite}
\usepackage{amsthm,mathtools,amssymb}
\usepackage{algorithm}
\usepackage{algorithmicx}
\usepackage[noend]{algcompatible}
\usepackage[inline]{enumitem}
\usepackage{xcolor}
\usepackage{subcaption}
\usepackage{multirow}
\usepackage{hyperref}
\hypersetup{
  colorlinks   = true,
  urlcolor     = black,    
  linkcolor    = black,    
  citecolor    = black      
}

\theoremstyle{definition}

\newtheorem{definition}{Definition}
\newtheorem{theorem}{Theorem}

\newtheorem{remark}{Remark}
\newtheorem{observation}{Observation}

\newcommand\NoDo{\renewcommand\algorithmicdo{}}

\newcommand\NoThen{\renewcommand\algorithmicthen{}}

\DeclarePairedDelimiter\ceil{\lceil}{\rceil}

\newcommand\blfootnote[1]{
  \begingroup
  \renewcommand\thefootnote{}\footnote{#1}
  \addtocounter{footnote}{-1}
  \endgroup}

\usepackage{pifont}

\newboolean{extendedVersion}
\setboolean{extendedVersion}{true}

\usepackage{balance}

\begin{document}

\ifthenelse{\boolean{extendedVersion}}{
\title{{\huge Knowledge Connectivity Requirements for Solving BFT Consensus with Unknown Participants and Fault Threshold\\
    (Extended Version)}}
}{
\title{{\huge Knowledge Connectivity Requirements for Solving BFT Consensus with Unknown Participants and Fault Threshold}}
}

\author{
    \IEEEauthorblockN{
        Hasan Heydari,
        Robin Vassantlal, and
        Alysson Bessani} 

        \IEEEauthorblockA{LASIGE, Faculdade de Ciências, Universidade de Lisboa, Portugal}
    
        \{hheydari, rvassantlal, anbessani\}@ciencias.ulisboa.pt
    }

\maketitle

\begin{abstract}
    Consensus is a fundamental building block for constructing reliable and fault-tolerant distributed services. 
    The increasing demand for high-performance and scalable blockchain protocols has brought attention to solving consensus in scenarios where each participant joins the system knowing only a subset of participants.
    In such scenarios, the participants' initial knowledge about the existence of other participants can collectively be represented by a directed graph known as \textit{knowledge connectivity graph}. 
    The \textit{Byzantine Fault Tolerant Consensus with Unknown Participants} (BFT-CUP) problem aims to solve consensus in those scenarios by identifying the necessary and sufficient conditions that the knowledge connectivity graphs must satisfy when a fault threshold is provided to all participants.
    This work extends BFT-CUP by eliminating the requirement to provide the fault threshold to the participants.
    We indeed address the problem of solving BFT consensus in settings where each participant initially knows a subset of participants, and although a fault threshold exists, no participant is provided with this information~-- referred to as \textit{BFT Consensus with Unknown Participants and Fault Threshold} (BFT-CUPFT).
    With this aim, we first demonstrate that the conditions identified for knowledge connectivity graphs by BFT-CUP are insufficient to solve BFT-CUPFT.
    Accordingly, we introduce a new type of knowledge connectivity graph that is sufficient for solving such a problem.
    To validate its sufficiency, we design a protocol for solving BFT-CUPFT.  
\end{abstract}


\ifbool{extendedVersion}{
    \blfootnote{This is a preprint of a paper to appear at the 44th IEEE International Conference on Distributed Computing Systems (ICDCS 2024).}
}{}

\section{Introduction}
\noindent\textbf{Context.} Consensus is a fundamental building block for constructing reliable and fault-tolerant distributed systems where participants agree on a common value out of the initially proposed values.
This problem has been primarily addressed in three settings~-- permissioned, permissionless, and hybrid.
In the permissioned setting, participants have a single global view of the system in advance, i.e., each participant is provided with the \textit{system's membership} and the \textit{fault threshold} (or, more generally, the fail-prone system~\cite{malkhi_1998}). 
Knowing such parameters simplifies the design and analysis of consensus protocols.

In the permissionless setting, protocols such as the Nakamoto consensus used in Bitcoin~\cite{nakamoto_2008} solve consensus without requiring a single global view of the system.
Specifically, no participant might be aware of the set of all participants.
Furthermore, the fault threshold is not explicitly defined in the same way as in the permissioned setting; instead, these protocols rely on assumptions about the overall distribution of resources within their networks, like computational power in Bitcoin.
Despite these protocols being scalable in terms of the number of participants, their performance is significantly lower by orders of magnitude compared to consensus protocols designed for the permissioned setting~\cite{vukolic_2015,korkmaz2022alder, Cachin2017BlockchainCP}.

The demand to scale consensus to accommodate numerous participants while maintaining high performance has led to the emergence of the hybrid setting. 
Protocols tailored for this setting relax the global view requirements found in permissioned consensus protocols, allowing each participant to have a partial view of the participants that it can trust or initially knows.
These protocols can be modeled in various ways. 
One approach is through the use of consensus with unknown participants (CUP)~\cite{cavin_2004, cavin_2005, greve_2007, alchieri_2008, alchieri_2016, vassantlal_2023}. Alternatively, federated Byzantine quorum systems~\cite{lokhava_2019,garcia_2018,mazieres_2015,garcia_2019}, personal Byzantine quorum systems~\cite{losa_2019}, quorum systems designed for permissionless networks~\cite{cachin_2022}, and heterogeneous quorum systems~\cite{li2023quorum} offer additional modeling options for hybrid consensus protocols.
This paper focuses on the CUP model.

\vspace{0.2em}
\noindent\textbf{The CUP model.}
In scenarios where each participant joins the network initially knowing only a subset of participants, the knowledge about the existence of other participants can collectively be represented by a directed graph known as \textit{knowledge connectivity graph}~\cite{cavin_2004}.
In such a graph, each vertex corresponds to a participant, and a directed edge from a vertex~$i$ to another vertex~$j$ denotes that participant~$i$ initially knows participant~$j$.
The graph depicted in Fig.~\ref{fig:osrExampleImpossible} is an example of a knowledge connectivity graph in which participant~$1$ initially knows participants~$2$,~$3$, and~$4$.

It is crucial to emphasize that \textit{the knowledge connectivity graph of a system might differ from its communication network}, as a participant~$i$ might be capable of communicating with another participant~$j$; however, such communication cannot happen when $i$ lacks knowledge about the existence of $j$ (i.e., there is no link from~$i$ to $j$ in their knowledge connectivity graph).

A knowledge connectivity graph must satisfy certain properties to allow participants to solve consensus. 
For instance, removing a faulty participant and its associated edges should not create multiple disconnected components, each having at least a correct participant.
This is because if the faulty participant remains silent during the execution, the correct participants in each disconnected component may decide on a value independently of the correct participants in other components.
This can result in multiple values being decided within the system, thereby failing to achieve consensus.
For example, in the graph of Fig.~\ref{fig:osrExampleImpossible}, in which only participant~$4$ is faulty, correct participants cannot solve consensus if participant~$4$ remains silent.

The CUP model determines the \textit{necessary} and \textit{sufficient} properties that a knowledge connectivity graph must satisfy to allow solving consensus under specific synchrony and fault assumptions. 
In this model, each participant is provided with the system's fault threshold and a subset of other participants when joining the network.
As synchrony and fault assumptions are relaxed, a participant's knowledge must be increased, i.e., the number of its outgoing edges in the knowledge connectivity graph must be increased.
For example, each participant in the Byzantine Fault-Tolerant~(BFT)~CUP~\cite{alchieri_2016} requires more knowledge than in the fault-free CUP model~\cite{cavin_2004}.   
The graph depicted in Fig.~\ref{fig:osrExamplePossible} illustrates an example of a knowledge connectivity graph that satisfies the requirements of the BFT-CUP model.

\begin{figure}[t!]
    \centering
    \begin{subfigure}[t]{0.48\textwidth}
        \centering
        \includegraphics[scale=0.8]{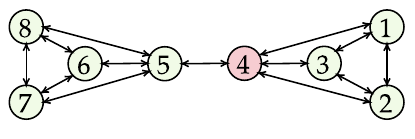}
        \caption{A knowledge connectivity graph that does not satisfy the requirements of the BFT-CUP model. 
        Although the number of Byzantine participants is less than one-third of the total participants, satisfying the requirement for solving the traditional Byzantine consensus~\cite{dwork_1988}, solving consensus in this system is impossible when participant~$4$ remains silent, as participants in $\{1, 2, 3\}$ cannot acquire knowledge about participants in $\{5, 6, 7, 8\}$, and vice versa.}
        \label{fig:osrExampleImpossible}
    \end{subfigure}
    \hfill
    \vspace{0.2em}
    \begin{subfigure}[t]{0.48\textwidth}
        \centering
        \includegraphics[scale=0.8]{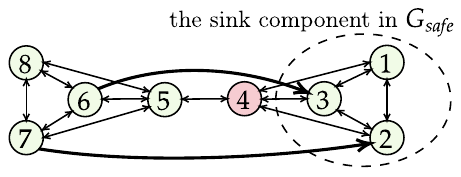}
        \caption{A knowledge connectivity graph that satisfies the requirements of the BFT-CUP model; hence, participants can solve consensus even in the presence of a Byzantine participant.}
        \label{fig:osrExamplePossible}
    \end{subfigure}
    \caption{Two knowledge connectivity graphs.
    In each graph, a vertex corresponds to a participant, and the outgoing edges from a vertex~$i$ point to the participants that $i$ initially knows; for instance, participant $1$ initially knows participants $2$, $3$, and~$4$.
    Participant~$4$ is Byzantine, and others are correct.}
    \label{fig:osrExamples}
\end{figure}

\vspace{0.2em}
\noindent\textbf{Problem statement.}
The BFT-CUP model solves consensus in partially synchronous systems, where each participant initially knows the system's fault threshold and a subset of participants.
Since knowing the system's fault threshold restricts the full potential of the BFT-CUP model in the hybrid setting, the main goal of this paper is to eliminate the explicit knowledge of that parameter by proposing a new model called the \textit{BFT Consensus with Unknown Participants and Fault Threshold} (BFT-CUPFT).
The BFT-CUPFT model is indeed an extension of BFT-CUP that enables participants to solve BFT consensus in partially synchronous systems where each participant initially knows only a subset of participants.

\vspace{0.2em}
\noindent\textbf{Contributions.}
To achieve such a goal, we begin by revisiting the BFT-CUP model under the extra assumption that each participant can use digital signatures 
(this assumption has been made in related work~\cite{garcia_2018,mazieres_2015,vassantlal_2023}).
We refer to the resulting model as the \textit{authenticated BFT-CUP} model.
Notably, Alchieri et al.~\cite{alchieri_2008} established that the requirements of knowledge connectivity graphs for solving consensus in the BFT-CUP model remain unchanged, irrespective of whether digital signatures are used.
Although such an assumption does not reduce the initial knowledge required by participants to solve consensus, it enables the design of a simpler consensus protocol.
We use the simpler protocol to highlight the significance of the fault threshold in the BFT-CUP model and clarify the necessary adjustments required when participants lack such information.

We then show an impossibility result stating that solving consensus is impossible in a partially synchronous system when the knowledge connectivity graph formed collectively by the initial knowledge of participants satisfies the requirements of the BFT-CUP model, yet no participant is provided with the fault threshold. 
This holds even when considering the aforementioned extra assumption.
As a result, we introduce the \textit{BFT Consensus with Unknown Participants and Fault Threshold} (BFT-CUPFT) model, defining a new type of knowledge connectivity graph that is sufficient to solve consensus without having information about the fault threshold.
To validate its sufficiency, we design a protocol for solving consensus in BFT-CUPFT.
Table~\ref{tbl} outlines the (im)possibility of solving BFT consensus deterministically under different system models.

In summary, the main contributions of this paper can be outlined as follows:
\begin{itemize}
    \item We revisit the BFT-CUP model by assuming each participant can use digital signatures.
    We design protocols that are simpler than the original ones under this new assumption.
    \item We prove an impossibility result for solving consensus when knowledge connectivity graphs satisfy the requirements of the BFT-CUP model, but participants are not provided with the fault threshold. 
    \item We introduce the BFT-CUPFT model by defining a new type of knowledge connectivity graph that allows participants to solve consensus when each participant initially knows only a subset of participants, not the system's membership or the fault threshold. 
    \item We design a consensus protocol that solves consensus in the BFT-CUPFT model.
\end{itemize}

\begin{table*}[t!]
\centering
\begin{tabular}{cl|cll|}
\cline{3-5}
\multicolumn{2}{c|}{\multirow{2}{*}{}} & \multicolumn{3}{c|}{\rule[-1ex]{0pt}{3.5ex}Knowledge about $n$ and $f$} \\ \cline{3-5} 
\multicolumn{2}{c|}{} & \multicolumn{1}{l|}{\begin{tabular}[c]{@{}l@{}}Known $n$, Known $f$\end{tabular}} & \multicolumn{1}{l|}{\rule[-1ex]{0pt}{3.5ex}\begin{tabular}[c]{@{}l@{}}Unknwon $n$, Known $f$\end{tabular}} & \begin{tabular}[c]{@{}l@{}}Unknown $n$, Unknown $f$\end{tabular} \\ \hline
\multicolumn{1}{|c|}{\multirow{3}{*}{\begin{tabular}[c]{@{}c@{}}Communication\\ Model\end{tabular}}} & Sync. & \multicolumn{1}{l|}{\rule[-1ex]{0pt}{3.5ex}\ding{51} (e.g.~\cite{abraham2017brief})} & \multicolumn{1}{l|}{\ding{51} (e.g.~\cite{khanchandani_2021})} & \multicolumn{1}{l|}{\ding{51} (e.g.~\cite{khanchandani_2021})} \\ \cline{2-5} 
\multicolumn{1}{|c|}{}  & \begin{tabular}[c]{@{}l@{}}\rule[-1ex]{0pt}{3.5ex}Partially Sync.\end{tabular} & \multicolumn{1}{l|}{\ding{51} (e.g.~\cite{castro_1999,hotstuff})} & \multicolumn{1}{l|}{\begin{tabular}[c]{@{}l@{}} \ding{51} (BFT-CUP~\cite{alchieri_2008,alchieri_2016})\end{tabular}} & \begin{tabular}[c]{@{}l@{}}{\color{teal} \ding{51} (BFT-CUPFT) $\leftarrow$ this work}\end{tabular} \\ \cline{2-5} 
\multicolumn{1}{|c|}{} & Async. & \multicolumn{1}{l|}{\rule[-1ex]{0pt}{3.5ex}\ding{55} (see~\cite{flp})} & \multicolumn{1}{l|}{\ding{55} (see~\cite{flp})} & \multicolumn{1}{l|}{\ding{55} (see~\cite{flp})} \\ \hline
\end{tabular}
\caption{The (im)possibility of solving Byzantine consensus deterministically under different system models.}
\label{tbl}
\end{table*}

\vspace{0.2em}
\noindent\textbf{Paper organization.}
The remainder of the paper is organized as follows. 
Section~\ref{sec:preliminaries} introduces our system model and provides the background for this paper.
Section~\ref{sec:revisiting:bft:cup} revisits the BFT-CUP model by assuming that each process can use digital signatures.
Section~\ref{sec:problem:specification} presents an impossibility result for solving consensus when knowledge connectivity graphs satisfy the requirements of the BFT-CUP model, but participants are not provided with the fault threshold.  
Section~\ref{sec:extended:bft:cup} introduces the BFT-CUPFT model, defining a new type of knowledge connectivity graph. 
Section~\ref{sec:consensus:extended:bft:cup} solves consensus in the BFT-CUPFT model.
Sections~\ref{sec:related:work} and~\ref{sec:conclusion} present the related work and conclude the paper, respectively.

\section{Preliminaries}\label{sec:preliminaries}
\subsection{System Model}
We consider a distributed system comprised of a finite set~$\Pi$ of processes operating under the assumption of partial synchrony~\cite{dwork_1988, makingByzConLive}, which guarantees that, for each execution of the protocol, there exists a time GST and a duration $\delta$ such that after GST, message delays between correct processes are bounded
by $\delta$. 
Before GST, messages may experience arbitrary delays.
We assume that each process~$i \in \Pi$ initially knows only a subset~$\Pi_{i} \subseteq \Pi$.

We denote the set of failed processes during an execution by $\Pi_F \subset \Pi$.
The faulty processes can behave arbitrarily,~i.e., can be Byzantine~\cite{lamport_1982}, and may collude and coordinate their actions.
A non-Byzantine process is said to be \textit{correct}.
We denote the set of correct processes during an execution by $\Pi_C = \Pi \setminus \Pi_F$.
We assume a static Byzantine adversary controlling the faulty processes, i.e., the set of processes controlled by the adversary is fixed at the beginning and does not change throughout the execution of the protocols.
We assume that correct processes neither know $\Pi$ nor $\Pi_F$; however, Byzantine processes may know both.

We further assume that each process has a unique ID, IDs are not necessarily consecutive, and it is infeasible for a faulty process to obtain additional IDs to launch a \emph{Sybil attack}~\cite{douceur_2002}.
Processes communicate by message passing through authenticated and reliable point-to-point channels, and when needed, can sign messages using digital signatures. 
We denote by $\langle m\rangle_i$ a message~$m$ signed by process~$i$.
A process $i$ can only send a message directly to another process $j$ if $i$ knows~$j$, i.e., if~$j \in \Pi_{i}$. 

\subsection{The Consensus Problem}\label{subsection:consensus} 
In the consensus problem, each process \emph{proposes} a value, and all correct processes must \emph{decide} the same value among the proposed values. 
Formally, any protocol that solves consensus must satisfy the following properties:
\begin{itemize}
    \item \emph{Validity:} if a correct process decides a value~$v$, then $v$ was proposed by some process.
    \item \emph{Agreement:} no two correct processes decide differently.
    \item \emph{Termination:} every correct process eventually decides some value.
    \item \emph{Integrity:} every correct process decides at most once.
\end{itemize}

\subsection{The BFT-CUP Model}
The \emph{Byzantine Fault-Tolerant Consensus with Unknown Participants} (BFT-CUP) model~\cite{alchieri_2016} solves consensus while tolerating Byzantine faults in partially synchronous systems.
The BFT-CUP model operates under the same system model as this paper, except it does not rely on digital signatures. 
It also assumes every process is provided with the fault threshold~$f$ as input (i.e.,~$|\Pi_F| \leq f$). 

In this model, each process initially obtains partial knowledge about the others using a local oracle called \emph{participant detector} (PD).
In further detail, let $\mathit{PD}_i$ represent the participant detector for process $i$, wherein $\mathit{PD}_i$ returns a set $\Pi_{i} \subseteq \Pi$ comprising the processes that $i$ can initially contact.
We assume that $\Pi_i$ can expand over time, but $\mathit{PD}_i$ always returns the same set.
Further, we assume that $\mathit{PD}_i$ might return~$\Pi$ if $i$ is Byzantine as by assumption, a Byzantine process may know~$\Pi$.

The information provided by the participant detectors of all processes collectively forms a directed graph known as \emph{knowledge connectivity graph}~\cite{cavin_2004}. 
Specifically, a directed graph $G_{\mathit{di}} = (V_{\mathit{di}}, E_{\mathit{di}})$ is a knowledge connectivity graph if~$V_{\mathit{di}}$ represents the set of processes $\Pi$, and an edge $(i, j) \in E_{\mathit{di}}$ exists if and only if process~$i$ initially knows process~$j$, i.e., $j \in \mathit{PD}_i$.
For instance, the graphs illustrated in Fig.~\ref{fig:osrExamples} exemplify knowledge connectivity graphs wherein $\mathit{PD}_1 = \{2,3,4\}$.
Since the PD of each process always returns the same set, the knowledge connectivity graph formed by PDs of processes is a \emph{static} graph. 

As stated before, the system's knowledge connectivity graph might differ from its communication network.
In further detail, in the BFT-CUP model, the communication network is assumed to be a complete graph as there is a reliable point-to-point channel between any two processes, according to the system model.
However, the knowledge connectivity graph is not necessarily a complete graph (e.g., see Fig.~\ref{fig:osrExamples}).
When $j \in \mathit{PD}_i$, process $i$ can send messages to $j$ from the beginning of the execution. 
However, when $j \notin \mathit{PD}_i$, process~$i$ first needs to discover $j$ through the processes it knows, and only after that it can directly communicate with~$j$.

The BFT-CUP model guarantees that processes can solve consensus when the knowledge connectivity graph meets specific requirements. 
Before delving into these requirements, we first review some relevant graph notations.
A directed graph $H_{\mathit{di}}=(V_{\mathit{di}},E_{\mathit{di}})$ is said to be \textit{$k$-strongly connected} if for any pair of nodes $i,j \in V_{\mathit{di}}$, $i$ can reach $j$ through at least~$k$ node-disjoint paths in $H_{\mathit{di}}$.
The \textit{strong connectivity} $\kappa(H_{\mathit{di}})$ of~$H_{\mathit{di}}$ is the maximum value of $k$ for which $H_{\mathit{di}}$ is $k$-strongly connected.
Given two non-empty sets of processes~$A$ and~$B$, we use the notation $A \xrightarrow[]{k} B$ to indicate that there are $k$ vertices in~$A$ that have outgoing neighbors in~$B$, where $k$ is a natural number.
Besides, $A \xrightarrow[]{>k} B$ denotes that the number of vertices in $A$ that have outgoing neighbors in $B$ is greater than $k$.

Given an directed graph $G_{\mathit{di}}=(V_{\mathit{di}},E_{\mathit{di}})$, we denote the subgraph induced by a set of nodes~$U_{\mathit{di}} \subseteq V_{\mathit{di}}$ as $G_{\mathit{di}}[U_{\mathit{di}}]$, i.e., $G_{\mathit{di}}[U_{\mathit{di}}] = \boldsymbol{(}U_{\mathit{di}}, \{ (i,j) \ | \ i, j \in U_{\mathit{di}} \land (i,j) \in E_{\mathit{di}}\}\boldsymbol{)}$. 
Further, we say a strongly connected component~$G_{\mathit{sink}}$ of $G_{\mathit{di}}$ is a \emph{sink} if and only if there is no path from a node in $G_{\mathit{sink}}$ to other nodes of $G_{\mathit{di}}$, except nodes in $G_{\mathit{sink}}$ itself.
A node $v\in V_{\mathit{di}}$ is a \emph{sink member} if it belongs to a sink component of $G_{\mathit{di}}$; otherwise, it is a \emph{non-sink member}.
An undirected counterpart can be defined for any directed graph~$G_{\mathit{di}}$ as $G = \boldsymbol{(}V_{\mathit{di}}, \{(i, j) \ | \ (i, j) \in E_{\mathit{di}} \lor (j, i) \in E_{\mathit{di}}\}\boldsymbol{)}$.
\begin{definition}[$k$-One Sink Reducibility ($k$-OSR) PD~\cite{alchieri_2016}]\label{def:osr}
A graph $G_{\mathit{di}}$ belongs to $k$-OSR PD if:
\begin{itemize}
    \item the undirected graph $G$ obtained from $G_{\mathit{di}}$ is connected,
    \item the directed acyclic graph obtained by reducing $G_{\mathit{di}}$ to its strongly connected components has exactly one sink, namely $G_\mathit{sink}=(V_\mathit{sink}, E_\mathit{sink})$,
    \item the sink component $G_\mathit{sink}$ is $k$-strongly connected, and
    \item there are at least $k$ node-disjoint paths from any process $i \notin V_\mathit{sink}$ to any process $j \in V_\mathit{sink}$.
\end{itemize}    
\end{definition} 
Given a knowledge connectivity graph $G_{\mathit{di}}=(V_{\mathit{di}},E_{\mathit{di}})$, we refer to $G_{\mathit{di}}[\Pi_C]$ by the \textit{safe} subgraph $G_{\mathit{safe}}$ of $G_{\mathit{di}}$.
With these definitions, we are ready to present the requirements that a knowledge connectivity graph must satisfy to enable solving BFT-CUP.

\begin{theorem}[\hspace{-0.01em}\cite{alchieri_2016}]\label{thm:bft:cup}
The safe subgraph $G_{\mathit{safe}}$ of a knowledge connectivity graph must satisfy the following two properties to ensure correct processes can solve BFT-CUP:
    \begin{itemize}
        \item $G_{\mathit{safe}}$ belongs to the $(f+1)$-OSR PD, and
        \item the sink component of $G_{\mathit{safe}}$ must contain at least $2f+1$ processes.
    \end{itemize}
\end{theorem}
If a knowledge connectivity graph meets the properties of Theorem~\ref{thm:bft:cup}, we say it satisfies the requirements of the BFT-CUP model.
We also denote the family of all finite graphs that meet these properties by $\mathcal{G}_{\mathit{di}}$.
Note that a graph that belongs to~$\mathcal{G}_{\mathit{di}}$ can have up to $f$ Byzantine nodes, but the properties specified in Theorem~\ref{thm:bft:cup} only describe the characteristics of the subgraph induced by correct nodes.
Fig.~\ref{fig:osrExamplePossible} is an example of a knowledge connectivity graph that satisfies the requirements of the BFT-CUP model in which participants in $\{1, 2, 3\}$ are the sink members of $G_{\mathit{safe}}$. 
Note that participant~$4$ is Byzantine and therefore not considered in $G_{\mathit{safe}}$.

\vspace{0.2em}
\noindent\textbf{Solving consensus in the BFT-CUP model.}
The sink component holds a pivotal role in the BFT-CUP model. 
In order to solve consensus within this model, processes expand their knowledge~-- i.e., the set of processes they initially know~-- and actively seek to identify the sink component.
Processes communicate using a communication primitive called reachable reliable broadcast~\cite{alchieri_2016} that allows delivery of a message if it is received through more than~$f$ node-disjoint paths from the sender.

Once the sink component is discovered, processes can establish intersecting quorums and execute a consensus protocol, such as PBFT~\cite{castro_1999}, on top of those quorums.
The definition of quorums highlights the significance of the sink: as demonstrated in~\cite{vassantlal_2023}, any defined quorum must include at least $\ceil{(|V_{\mathit{sink}}|+f+1)/2}$ sink processes to intersect with any other quorum in at least one correct process.
\begin{remark}
    The BFT-CUP model does not rely on digital signatures.
\end{remark}

\section{Revisiting the BFT-CUP Model}\label{sec:revisiting:bft:cup}
This section revisits the BFT-CUP model, incorporating an additional assumption that allows each process to use digital signatures~-- denoting the resulting model as the \textit{authenticated BFT-CUP} model.
Recall that in the BFT-CUP model, a process delivers a message if it is received through more than $f$ node-disjoint paths.
However, using digital signatures, processes can deliver messages without waiting to receive them through multiple paths.
This significantly simplifies the communication protocols and reduces the places where the fault threshold is required.
As presented next, the consensus protocol in the authenticated BFT-CUP model has roughly 20 lines compared with 120 in \cite{alchieri_2016}.
Although the requirements of the knowledge connectivity graphs to solve consensus remain unchanged~\cite{alchieri_2008}, in the simplified protocol, we can highlight where the fault threshold is still required and clarify the necessary adjustments to remove it entirely.

\subsection{Consensus Protocol in the Authenticated BFT-CUP Model}
In order to solve consensus in the authenticated BFT-CUP model, we introduce three algorithms~-- namely Discovery, Sink, and Consensus~-- along with their associated properties. 
The relationship between these algorithms is as follows: the Consensus algorithm executes the Sink algorithm, and the Sink algorithm executes the Discovery algorithm.

\vspace{0.2em}
\noindent\textbf{Discovery algorithm.}
The Discovery algorithm, described in Algorithm~\ref{alg:discovery:known:f}, enables any correct process to expand the set of processes that it knows with the aim of eventually receiving the PD of each correct process reachable from it. 
This algorithm provides only one task, $\mathtt{discovery}$.
When a process executes this task, it periodically asks the processes it knows to respond by sending the PDs they have received. 
\begin{algorithm}[!t]
\caption{The Discovery algorithm -- process $i$.}
\label{alg:discovery:known:f}
\begin{algorithmic}[1]
\NoThen\NoDo
    \STATEx{\hspace{-1.67em}\textbf{task} $\mathtt{discovery}()$}
        \STATE{$\mathcal{S}_\mathit{PD}^i,\mathcal{S}_\mathit{known}^i,\mathcal{S}_\mathit{received}^i \leftarrow \{ \langle i, \mathit{PD}_i \rangle_i \}, \mathit{PD}_i \cup \{i\}, \{i\}$}\label{line:sign}
        \STATE{\textbf{periodically} $\forall j \in \mathcal{S}_\mathit{known}^i:$ \textbf{send} $\langle \textsc{GetPDs} \rangle$ \textbf{to} $j$}\label{line:get:pds}

    \vspace{0.2em}
    \STATEx{\hspace{-1.67em}\textbf{upon receiving} $\langle \textsc{GetPDs} \rangle$ \textbf{from} $j$}
    \STATE{\textbf{send} $\langle \textsc{SetPDs}, \mathcal{S}_\mathit{PD}^i \rangle$ \textbf{to} $j$}\label{line:send:PDs}

    \vspace{0.2em}
    \STATEx{\hspace{-1.67em}\textbf{upon receiving} $\langle \textsc{SetPDs}, \mathcal{S}_\mathit{PD}^j \rangle$ \textbf{from} $j$}  
        \STATE{$\mathcal{S}_\mathit{PD}^i \leftarrow \mathcal{S}_\mathit{PD}^i \cup \mathcal{S}_\mathit{PD}^j$}\label{line:start:update}
        \STATE{$\mathcal{S}_\mathit{known}^i \leftarrow \mathcal{S}_\mathit{known}^i \cup \{ k \in \mathit{PD}_* \ | \ \langle *, \mathit{PD}_*\rangle_* \in \mathcal{S}_\mathit{PD}^j \}$}
        \STATE{$\mathcal{S}_\mathit{received}^i \leftarrow \mathcal{S}_\mathit{received}^i \cup \{ k \ | \ \langle k, *\rangle_* \in \mathcal{S}_\mathit{PD}^j \}$}\label{line:end:update}
    \end{algorithmic}
\end{algorithm}
Each process $i$ has the following three local sets:
\begin{itemize}
    \item $\mathcal{S}_\mathit{PD}^i$ -- Process $i$ stores any received PD in this set, initialized with $\{\langle i, \mathit{PD}_i \rangle_i\}$.
    \item $\mathcal{S}_\mathit{known}^i$ -- This set contains the processes that $i$ knows, initialized with $\mathit{PD}_i \cup \{i\}$.
    \item $\mathcal{S}_\mathit{received}^i$ -- This set contains the set of processes that $i$ has received their PDs, initialized with $\{i\}$.
\end{itemize}

Periodically, each process~$i$ sends a $\langle \textsc{GetPDs} \rangle$ message to all processes that it knows, requesting each of them to share its $\mathcal{S}_\mathit{PD}^*$ (line~\ref{line:get:pds}).
Upon receiving a $\langle \textsc{GetPDs} \rangle$ request from process~$j$, process~$i$ responds by sending the PDs it knows, i.e., sending set $\mathcal{S}_\mathit{PD}^i$ to $j$ (line~\ref{line:send:PDs}). 
When $i$ receives a message $\langle \textsc{SetPDs}, \mathcal{S}_\mathit{PD}^j \rangle$ from process $j$, it updates its local sets using $\mathcal{S}_\mathit{PD}^j$ (lines~\ref{line:start:update}-\ref{line:end:update}).
It is worth noting that, since correct processes sign their PDs (line~\ref{line:sign}), Byzantine processes cannot lie about the PD of any correct process~$i$, either by modifying $\mathit{PD}_i$ or by creating a PD for~$i$.

This algorithm satisfies two properties that are specified in the following theorem.
\begin{theorem}\label{thm:discovery:sink}
    Consider a system with a knowledge connectivity graph $G_{\mathit{di}} \in \mathcal{G}_\mathit{di}$. 
    Assuming $V_{\mathit{sink}}$ comprises the sink members of $G_{\mathit{di}}$, by executing Algorithm~\ref{alg:discovery:known:f} in this system, every correct process eventually
    \begin{enumerate*}[label=(\alph*)]
        \item discovers all correct sink members, i.e., 
        $\forall i \in \Pi_C: V_{\mathit{sink}} \cap \Pi_C \subseteq \mathcal{S}_\mathit{known}^i$, and
        \item receives the PDs of all correct sink members, i.e.,
        $\forall i \in \Pi_C: V_{\mathit{sink}} \cap \Pi_C \subseteq \mathcal{S}_\mathit{received}^i$.
    \end{enumerate*}
\end{theorem}
\ifbool{extendedVersion}{
The proof of the aforementioned theorem can be found in Appendix~\ref{appendix:discovery}.
}{
The proof of the aforementioned theorem, along with proofs of other theorems not presented here, can be found in the extended version of this paper~\cite{extended}.
}

\vspace{0.2em}
\noindent\textbf{Sink algorithm.}
The Sink algorithm enables each process to discover a sink component.
Before introducing this algorithm, we present two theorems for a knowledge connectivity graph belonging to $\mathcal{G}_{\mathit{di}}$.
The first theorem ensures the existence of two sets with specific properties, while the second theorem demonstrates that the union of these sets contains all and only sink members.

\begin{theorem}\label{thm:sink:alternative:def}
    In a knowledge connectivity graph $G_{\mathit{di}}=(V_{\mathit{di}},E_{\mathit{di}})$ belonging to $\mathcal{G}_{\mathit{di}}$, there are two sets $S_1,S_2 \subseteq V_{\mathit{di}}$ that satisfy the following properties for a given fault threshold~$f$:
    \begin{enumerate}[label=P\arabic*)]
        \item The size of $S_1$ is greater than or equal to $2f+1$, i.e., $|S_1|\geq 2f+1$.
        \item The strong connectivity of the subgraph induced by $S_1$ is greater than or equal to $f+1$, i.e., $\kappa(G_{\mathit{di}}[S_1]) \geq f+1$.
        %
        \item The number of processes in $S_1$ that have outgoing edges to ~$V_{\mathit{di}}\setminus S_1$ is at most $f$, i.e., $S_1 \xrightarrow[]{\leq f} V_{\mathit{di}} \setminus S_1$.
        \item $S_2$ contains any process~$i \notin S_1$ such that the number of processes in $S_1$ that have outgoing edges to~$i$ is greater than~$f$, i.e., $\forall i \in V_{\mathit{di}}\setminus S_1 : S_1 \xrightarrow[]{> f}~\{i\} \iff i \in S_2$.
    \end{enumerate}
\end{theorem}

    

\ifbool{extendedVersion}{
The proof of the aforementioned theorem and other theorems related to the Sink algorithm can be found in Appendix~\ref{appendix:sink}.
}{}

\begin{theorem}\label{thm:sink:exists}
    In a knowledge connectivity graph $G_{\mathit{di}}=(V_{\mathit{di}},E_{\mathit{di}})$ belonging to $\mathcal{G}_{\mathit{di}}$, if a fault threshold~$f$ and two subsets of processes~$S_1$ and~$S_2$ satisfy the properties specified in Theorem~\ref{thm:sink:alternative:def}, then set~$S_1 \cup S_2$ contains all and only the sink members.
\end{theorem}

In Theorems~\ref{thm:sink:alternative:def} and~\ref{thm:sink:exists}, set~$S_1$ (resp.~$S_2$) comprises the subset of sink members whose strong connectivity can (resp. cannot) be calculated.
To clarify why the strong connectivity of~$S_2$ cannot be calculated, we determine sets~$S_1$ and~$S_2$ in two scenarios.
In these scenarios, we assume that each sink member executes the Discovery algorithm.
For simplicity, we assume there are~$f$ Byzantine processes in the sink.
This happens when each correct sink member has at least~$f+1$ node-disjoint paths to each Byzantine sink member, i.e., for each Byzantine sink member~$i$, there are at least~$f+1$ correct sink members that know~$i$.
\begin{itemize}
    \item \textbf{Scenario I)} Byzantine sink members remain silent during the execution, and every correct sink member receives the PDs of other correct sink members by time~$t \geq 0$.
    At time~$t$,~$S_1$ contains all correct sink members while~$S_2$ contains all Byzantine sink members.
    Note that each correct sink member can calculate the strong connectivity of~$S_1$ at time~$t$, while it cannot calculate the strong connectivity of~$S_2$ or any other set containing at least a Byzantine sink member.
    This limitation arises because for computing the strong connectivity of a set, it is required to have the PDs of all processes in that set. 
    \item \textbf{Scenario II)} By time~$t$, each sink member receives the PDs of other sink members, except for $f$ correct sink members~$D$. 
    Since a correct process that is not a member of~$D$ cannot distinguish between this and the previous scenario, it determines sets~$S_1$ and~$S_2$ at time~$t$, considering the possibility that members of~$D$ might be Byzantine and stay silent.
    Accordingly,~$S_1$ contains those sink members that are not in set~$D$ while $S_2=D$.
    Like the previous case, as members of $S_1$ have not received the PDs of members of $S_2$, they cannot calculate the strong connectivity of $S_2$ or any set that contains at least a member of~$D$.
\end{itemize}

The following steps provide an insight into why the properties presented in Theorem~\ref{thm:sink:exists} identify a sink:
\begin{itemize}
    \item Since $|\Pi_F| \leq f$, Property P1 implies that $S_1$ contains at least $f+1$ correct processes.
    \item Since $G_\mathit{di} \in \mathcal{G}_{\mathit{di}}$, there is no link from a correct sink member to a correct non-sink member.
    Hence, $S_1$ cannot contain both correct sink and correct non-sink members due to Property~P2.
    Accordingly, all correct members of~$S_1$ are either sink members or non-sink members.
    \item For the sake of contradiction, assume that all correct members of~$S_1$ are non-sink members.
    Since $G_\mathit{di} \in \mathcal{G}_{\mathit{di}}$, any correct non-sink member has at least $f+1$ node-disjoint paths to correct sink members.
    Therefore, all correct members of $S_1$ must have at least~$f+1$ node-disjoint paths to the sink members.
    This implies that there are at least~$f+1$ outgoing edges from~$S_1$ to the remaining processes. 
    Consequently,~$S_1$ cannot satisfy Property~P3, which means all correct members of~$S_1$ are sink members.
    \item For the sake of contradiction, assume that $S_1 \cup S_2$ is a proper subset of a sink, i.e., there is at least a process~$i \in V_{\mathit{sink}}$ such that $i \notin S_1\cup S_2$.
    Since $G_\mathit{di} \in \mathcal{G}_{\mathit{di}}$, there are at least $f+1$ node-disjoint paths from members of $S_1$ to~$i$.
    Accordingly,~$i$ must be a member of~$S_2$, which is a contradiction.
    Hence, $S_1 \cup S_2$ contains all members of a sink.
\end{itemize}

We define a predicate $\mathtt{isSink}_{G_{\mathit{di}}}$ to check whether a given fault threshold~$f$ and two sets of processes~$S_1$ and~$S_2$ satisfy the properties specified in Theorem~\ref{thm:sink:alternative:def}.
Specifically, $\mathtt{isSink}_{G_{\mathit{di}}}(f,S_1,S_2)$ is \textit{true} if and only if these parameters meet the properties of such a theorem.

Taking into account the existence of Byzantine processes in the sink, the Sink algorithm (Algorithm~\ref{alg:sink:known:f}) enables each process~$i$ to discover the sink members when the knowledge connectivity graph~$G_{\mathit{di}} \in \mathcal{G}_\mathit{di}$.
This algorithm provides only one function, $\mathtt{sink}$.
Using this function, each correct process continuously expands the set of processes it knows by executing the Discovery algorithm until it identifies the sink component.
It terminates the Sink algorithm by returning a set equal to $V_{\mathit{sink}}$ when there exists two sets $S_1$ and $S_2$ with conditions specified in line~\ref{line:sink:condition}.
These conditions are analogous to the properties specified in Theorem~\ref{thm:sink:alternative:def}.

\begin{algorithm}[t!]
\caption{The Sink algorithm -- process $i$.}
\label{alg:sink:known:f}
\begin{algorithmic}[1]
\NoThen\NoDo

    \STATEx{\hspace{-1.67em}\textbf{function} $\mathtt{isSink}_{G_{\mathit{di}}}(f, S_1, S_2)$}
    \STATE{\textbf{return} $|S_1|\geq 2f+1 \land \kappa(S_1) \geq f+1 \land (S_1 \xrightarrow[]{\leq f} \mathcal{S}_\mathit{known}^i \setminus S_1) \land S_2 = \{ j \ | \ \forall j \in \mathcal{S}_\mathit{known}^i \setminus S_1 : S_1 \xrightarrow[]{> f} \{j\} \}$}   
    
    \vspace{0.2em}
    \STATEx{\hspace{-1.67em}\textbf{function} $\mathtt{sink}()$}
    \STATE{\textbf{fork} $\mathtt{discovery}()$}

    \STATE{\textbf{wait until} $\exists S_1\subseteq \mathcal{S}_\mathit{received}^i, \exists S_2\subseteq \mathcal{S}_\mathit{known}^i \setminus S_1 :$}\label{line:sink:condition}
    \STATEx{$\quad \mathtt{isSink}_{G_{\mathit{di}}}(k-1, S_1, S_2) = \mathit{true}$}
    \STATE{\textbf{return} $S_1 \cup S_2$}
\end{algorithmic}
\end{algorithm}

Using an example, we demonstrate why the Sink algorithm returns $S_1 \cup S_2$ as the sink.
Consider the knowledge connectivity graph depicted in Fig.~\ref{fig:osrExamplePossible}.
Suppose process $1$ executes the $\mathtt{sink}$ function, and consider the following scenario.
Process~$2$ is slow, and process $4$ sends a set $P=\{1,2,3\}$ as its PD during the execution of $\mathtt{discovery}$.
When process $1$ receives~$P$ and~$\mathit{PD}_3$, the conditions specified in line~\ref{line:sink:condition} of Algorithm~\ref{alg:sink:known:f} are satisfied with $S_1=\{1,3,4\}$.
Note that there are more than $f$ processes in $S_1$ that have outgoing edges to process~$2$ so $S_2=\{2\}$.
Consequently, set $S_1\cup S_2=\{1,2,3,4\}$ is returned as the sink.
It is worth noting that in this algorithm, the connectivity of $S_1$ can be computed as $S_1 \subseteq \mathcal{S}_\mathit{received}^*$, while the connectivity of $S_2$ cannot be computed since the PDs of processes in $S_2$ were not received.

The Sink algorithm satisfies the properties specified in the following theorem.
\begin{theorem}\label{thm:sink:termination}
    In a knowledge connectivity graph $G_{\mathit{di}} \in \mathcal{G}_{\mathit{di}}$, Algorithm~\ref{alg:sink:known:f} executed by any correct process 
    \begin{enumerate*}[label=(\alph*)]
        \item eventually terminates and 
        \item returns all sink members.
    \end{enumerate*}
\end{theorem}

\vspace{0.2em}
\noindent\textbf{Solving consensus in the authenticated BFT-CUP model.}
Algorithm~\ref{alg:consensus} solves consensus in the authenticated BFT-CUP model.
This algorithm provides a $\mathtt{propose}$ function through which processes can propose a value and decide on a common value.
In this algorithm, each process $i$ first executes the Sink algorithm.
Upon the termination of the Sink algorithm, process~$i$ acts based on whether it is a sink or a non-sink member.
If it is a sink member, it executes a traditional consensus protocol (e.g., PBFT~\cite{castro_1999}) with the sink members.
Otherwise, it asks the sink members to respond by sending the decided value and waits until receiving the same value~$v$ from $\ceil{(|S|+1)/2}$ distinct sink members, where $S$ comprises the sink members.
Since the sink component contains at least $2f+1$ correct and at most $f$ Byzantine processes, $\ceil{(|S|+1)/2}\geq f+1$.
This implies that there is at least one correct process among the answering processes.
Thus, a correct process can be sure that~$v$ is the value decided by correct sink members and, therefore, can decide~$v$.

\begin{algorithm}[t!]
\caption{The Consensus algorithm -- process $i$.}
\label{alg:consensus}
\begin{algorithmic}[1]
\NoThen\NoDo
    \STATEx{\hspace{-1.67em}\textbf{variables}}
    \STATE{$\mathit{val} \leftarrow \perp$}

    \vspace{0.2em}
    \STATEx{\hspace{-1.67em}\textbf{function} $\mathtt{propose}(v)$}
    \STATE{$S \leftarrow \mathtt{sink}()$}
    \IF{$i \in S$}
        \STATE{$\mathit{val} \leftarrow \mathtt{consensus}.\mathtt{propose}(S,v)$}
    \ELSE
        \STATE{$\forall j \in S: $ \textbf{send} $\langle \textsc{GetDecidedVal} \rangle$ \textbf{to} $j$}
        \STATE{\textbf{wait until receiving the same} $\langle \textsc{DecidedVal}, \mathit{val} \rangle$}
        \STATEx{\qquad \textbf{from} $\ceil{(|S|+1)/2}$ distinct processes in $S$}
    \ENDIF
    \STATE{\textbf{return} $\mathit{val}$}

    \vspace{0.2em}
    \STATEx{\hspace{-1.67em}\textbf{upon receiving} $\langle \textsc{GetDecidedVal} \rangle$ \textbf{from} $j$}
    \STATE{\textbf{wait until} $\mathit{val} \neq \perp$}
    \STATE{\textbf{send} $\langle \textsc{DecidedVal}, \mathit{val} \rangle$ \textbf{to} $j$}
\end{algorithmic}
\end{algorithm}

\begin{theorem}\label{thm:consensus}
    Algorithm~\ref{alg:consensus} solves consensus in the authenticated BFT-CUP model.
\end{theorem}
\ifbool{extendedVersion}{
The proof of the aforementioned theorem can be found in Appendix~\ref{appendix:consensus}.
}{}

\subsection{The Fault Threshold's Role in the BFT-CUP Model}
In the three algorithms presented to solve consensus in the authenticated BFT-CUP model, the only place where the fault threshold is used is in the Sink algorithm (Algorithm~\ref{alg:sink:known:f}).
In that algorithm, each process requires knowing~$f$ in order to identify the sink and terminate the algorithm.
Consequently, the lack of access to $f$ may lead to the following issues:
\begin{enumerate*}[label=(\alph*)]
    \item the Sink algorithm does not terminate,
    \item multiple subsets of the sink might declare themselves as the sink, or
    \item a subset of non-sink members might declare themselves as the sink.
\end{enumerate*}
As processes must discover the sink to solve consensus, the Termination property of consensus is violated if the first scenario happens.
If the second and third scenarios happen, the Agreement property of consensus can be violated as members of each sink can solve consensus independently of other processes.

\section{An Impossibility Result}\label{sec:problem:specification}
The primary objective of this paper is to extend the BFT-CUP model to enable solving consensus when each process only knows its PD but does not know $f$.
We achieve this by finding the sufficient knowledge connectivity requirements for solving consensus in this setting.

The initial step for finding such requirements involves addressing the following question:
\emph{in partially synchronous systems, is having a knowledge connectivity graph that satisfies the requirements of the BFT-CUP model sufficient to solve consensus when the fault threshold is unknown?} 
That is, in a partially synchronous system, can processes solve consensus when each process~$i$ initially has access to its participant detector $\mathit{PD}_i$, the sets returned by participant detectors of processes collectively form a knowledge connectivity graph $G_{\mathit{di}} \in \mathcal{G}_{\mathit{di}}$, and no correct process knows the value of $f$?
Note that the system still has a fault threshold as required to define~$G_{\mathit{di}}$, but it is unknown.
We negatively answer the above question by presenting a theorem that relies on the indistinguishability technique (e.g., used in ~\cite{flp, waitFree, heydari_hard}).

\begin{figure}[t!]
    \centering
    \begin{subfigure}[t]{0.22\textwidth}
        \centering
        \includegraphics[scale=0.8]{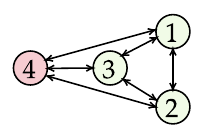}
        \caption{System $A$: a $2$-OSR PD in which only process $4$ is faulty.}
        \label{fig:imp:a}
    \end{subfigure}%
    \hfill
    \begin{subfigure}[t]{0.22\textwidth}
        \centering
        \includegraphics[scale=0.8]{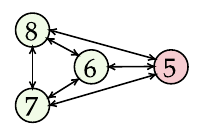}
        \caption{System $B$: a $2$-OSR PD in which only process $5$ is faulty.}
        \label{fig:imp:b}
    \end{subfigure}
    \\
    \begin{subfigure}[t]{0.45\textwidth}
        \centering
        \includegraphics[scale=0.8]{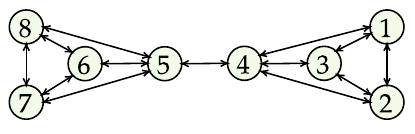}
        \caption{System $\mathit{AB}$: a $1$-OSR PD in which all processes are correct.}
        \label{fig:imp:c}
    \end{subfigure}
    
    \caption{Processes in $\{1,2,3\}$ cannot distinguish between case~a (if process~$4$ remains silent) and case~c (if process~$4$ is slow). 
    Likewise, Processes in $\{6,7,8\}$ cannot distinguish cases b and~c.}
    \label{fig:imp}
\end{figure}

\begin{theorem}\label{thm:kosr:impossibility}
    In partially synchronous systems, a knowledge connectivity graph belonging to~$\mathcal{G}_\mathit{di}$ is insufficient to solve Byzantine consensus when the fault threshold is unknown. 
\end{theorem}
\begin{proof}
    We establish the theorem for a weaker failure model, specifically crash faults, by assuming that a process fails by crashing. 
    Given that any impossibility result derived for a weaker model holds for a stronger model, an impossibility result drawn for crash faults is also valid for Byzantine faults.
    
    For the sake of contradiction, assume that there is a protocol~$\mathcal{A}$ by which processes can solve consensus when the knowledge connectivity graph of the system belongs to $\mathcal{G}_\mathit{di}$, but no process knows the value of the fault threshold.
    We present three cases, each with a corresponding knowledge connectivity graph. 
    It is straightforward to validate that each knowledge connectivity graph belongs to $\mathcal{G}_\mathit{di}$.
    \begin{enumerate}[label=\alph*), leftmargin=1.3em]
        \item Consider a distributed system $A$ composed of a set of processes $\Pi_A=\{1,2,3,4\}$ with the knowledge connectivity graph depicted in Fig.~\ref{fig:imp:a}.
        Except for process $4$, other processes are correct.    
        Assume that the initial value of every process is $v$.
        Additionally, assume that the GST occurs at most by time $t_A$ in this system (in accordance with the definition of partial synchrony, making such an assumption is possible).
        By assumption, processes $1,2,$ and $3$ must be able to solve consensus using $\mathcal{A}$.
        Let~$E_A$ be an execution of $\mathcal{A}$ with duration $\Delta_A$, resulting in deciding~$v$ due to the Validity property of consensus.
    
        \item Similar to the previous case, consider a distributed system~$B$ composed of a set of processes $\Pi_B=\{5,6,7,8\}$ with the knowledge connectivity graph depicted in Fig.~\ref{fig:imp:b}.
        Except for process $5$, other processes are correct.    
        Assume that the initial value of every process is $u$.
        Additionally, assume that the GST occurs at most by time $t_B$ in this system.
        By assumption, processes $6,7,$ and~$8$ must be able to solve consensus using $\mathcal{A}$.
        Let $E_B$ be an execution of~$\mathcal{A}$ with duration $\Delta_B$, resulting in deciding~$u$ due to the Validity property of consensus.   
    
        \item Consider a distributed system $\mathit{AB}$ composed of eight processes $\Pi_A\cup \Pi_B = \{1,2,\dots, 8\}$ with the knowledge connectivity graph depicted in Fig.~\ref{fig:imp:c}.
        Assume that all processes are correct and the initial value of each member of $\Pi_A$ (resp. $\Pi_B$) is $v$ (resp. $u$). 
        Furthermore, assume that the communication delays between any two members of $\{1,2,3\}$ (resp. $\{6,7,8\}$) are the same as the communication delays in system~$A$ (resp. $B$). 
        However, any message sent between any other two processes will be received after $\mathtt{max}\{t_A+\Delta_{A},t_B+ \Delta_{B}\}$.
        Note that processes $1,2,$ and $3$ cannot distinguish cases a and c, so they must decide~$v$.
        Likewise, processes $6,7,$ and $8$ cannot distinguish cases~b and~c, so they must decide~$u$.
        Thus, the Agreement property is violated.
    \end{enumerate}
    
    The violation of the Agreement property in the third case implies a contradiction. 
    Consequently, our assumption that states there is a protocol by which processes can solve consensus when the fault threshold is unknown in a system with a knowledge connectivity graph belonging to~$\mathcal{G}_\mathit{di}$ is incorrect, completing the proof.
\end{proof}

The above theorem describes executions where, before deciding a value, no process in~$\{1,2,3\}$ can discover a process in~$\{6,7,8\}$ and vice versa.
Note that in system~$\mathit{AB}$, processes in~$\{1,2,3\}$ (resp.~$\{6,7,8\}$) can consider themselves as sink members since they are sink members in system~$A$ (resp. system~$B$), and they cannot make a distinction between systems~$A$ and $\mathit{AB}$ (resp. systems~$B$ and $\mathit{AB}$).
That is when processes do not know $f$, since $\mathtt{isSink}_{G_{\mathit{di}}}(1, \{1,2,3\}, \{4\}) = \mathit{true}$, and $\mathtt{isSink}_{G_{\mathit{di}}}(1, \{6,7,8\}, \{5\}) = \mathit{true}$, sets $\{1,2,3,4\}$ and $\{5,6,7,8\}$ are sinks.
This result gives rise to the following observation:

\begin{observation}\label{observation1}
    In partially synchronous systems, when the fault threshold is unknown, and the knowledge connectivity graph belongs to~$\mathcal{G}_\mathit{di}$,
    if there exists a natural number $g$ and two subsets of processes~$S_1$ and $S_2$ such that $\mathtt{isSink}_{G_{\mathit{di}}}(g, S_1, S_2) = \mathit{true}$, then processes in $S_1$ consider $S_1 \cup S_2$ as the sink members. 
\end{observation}

Note that a subset of non-sink members might also declare themselves as a sink.
For example, in Fig.~\ref{fig:imp:ns:b}, processes in $\{1, 2,3,4,6\}$ that are non-sink members can declare themselves as a sink.
Specifically, $\mathtt{isSink}_{G_{\mathit{di}}}(2, \{1, 2,3,4,6\}, \{5,7\}) = \mathit{true}$.
If that happens, the Agreement property of consensus can be violated as members of each sink can solve consensus independently of other processes.

\begin{figure}[t!]
    \centering
    \begin{subfigure}[t]{0.26\textwidth}
        \centering
        \includegraphics[scale=0.8]{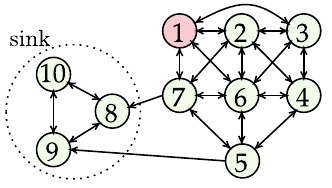}
        \caption{System $A$: a 2-OSR PD in which only process $1$ is faulty.}
        \label{fig:imp:ns:b}
    \end{subfigure}
    \hfill
    \begin{subfigure}[t]{0.2\textwidth}
        \centering
        \includegraphics[scale=0.8]{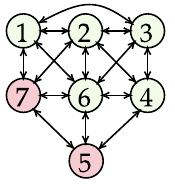}
        \caption{System $B$: a 3-OSR PD in which processes $5$ and $7$ are faulty.}
        \label{fig:imp:ns:a}
    \end{subfigure}%
    
    \caption{Processes in $\{2,3,4,6\}$ cannot distinguish between case~a (if processes $5$ and $7$ remain silent) and case~b (if process $1$ behaves like correct processes but processes~$5$ and~$7$ are slow).}
    \label{fig:imp:non:sink}
\end{figure}

This impossibility result indicates that a new type of knowledge connectivity graph is required for processes to solve consensus using only their PDs and without knowing~$f$. 
In the next section, we achieve this by modifying a knowledge connectivity graph that meets the BFT-CUP model's requirements by adding new edges or removing existing ones.


\section{The BFT-CUPFT Model}\label{sec:extended:bft:cup}
The primary challenge in solving consensus in the BFT-CUPFT model, when the initial knowledge of processes collectively forms a knowledge connectivity graph belonging to~$\mathcal{G}_{\mathit{di}}$, lies in the possibility of 
\begin{enumerate*}[label=(\alph*)]
    \item existing multiple subsets of processes, each identifying itself as a sink, and
    \item having some correct processes that cannot identify whether they are sink or non-sink members, thereby they cannot solve consensus.
\end{enumerate*}
Therefore, to enable processes to solve consensus in the BFT-CUPFT model, it is crucial to prevent the emergence of multiple sinks, along with ensuring that each process can identify whether it is a sink or non-sink member.
This objective is accomplished in this section by adding new edges to, or removing the existing ones from, the knowledge connectivity graphs belonging to~$\mathcal{G}_{\mathit{di}}$, thereby creating a new type of knowledge connectivity graph.
Before introducing such a graph, we present a few definitions.

\vspace{0.2em}
\noindent\textbf{Defining a sink without a known fault threshold.}
Recall from Observation~\ref{observation1} that when processes lack information about the fault threshold, if there exists a natural number $g$, along with two subsets of processes~$S_1$ and~$S_2$, such that $\mathtt{isSink}_{G_{\mathit{di}}}(g, S_1, S_2) = \mathit{true}$, then processes in~$S_1$ consider set~$S_1 \cup S_2$ as members of a sink. 
When the fault threshold is unknown, we define a subset of processes~$S$ as a sink if:
\begin{align*}
    &\mathtt{isSink}_{G_{\mathit{di}}}^*(S) \iff \exists g \geq 0, \exists S_1, S_2 \subseteq S: S_1 \cup S_2 = S \ \land 
    \\& \quad \mathtt{isSink}_{G_{\mathit{di}}}(g, S_1, S_2).
\end{align*}

When a set~$S = S_1 \cup S_2$ is identified as a sink, we denote $\mathtt{f}_{G_{\mathit{di}}}(S)$ as the maximum value of $g$ that satisfies $\mathtt{isSink}_{G_{\mathit{di}}}(g, S_1, S_2)$.
Additionally, $\mathtt{k}_{G_{\mathit{di}}}(S)=\mathtt{f}_{G_{\mathit{di}}}(S)+1$ represents the connectivity of $S$.

\begin{figure}[t!]
    \centering
    \begin{subfigure}[t]{0.44\textwidth}
        \centering
        \includegraphics[scale=0.77]{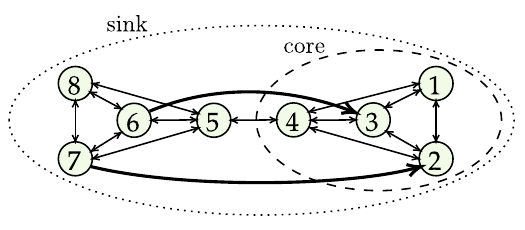}
        \caption{In this knowledge connectivity graph, the sink component differs from the core component.
        }
        \label{fig:rich:k:osr}
    \end{subfigure}
    \\
    \begin{subfigure}[t]{0.44\textwidth}
        \centering
        \includegraphics[scale=0.8]{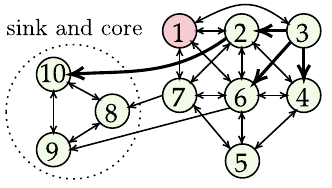}
        \caption{In this knowledge connectivity graph, the sink component is the same as the core component.}
        \label{fig:rich:k:osr:2}
    \end{subfigure}
    
    \caption{Two knowledge connectivity graphs that satisfy the requirements of the BFT-CUPFT model.}
    \label{fig:rich:k:osr:all}
\end{figure}

\vspace{0.2em}
\noindent\textbf{New type of knowledge connectivity graph.}
We define a new type of knowledge connectivity graph in which multiple subsets of processes might consider themselves as sinks.
However, only members of a distinguished sink called the \textit{core} can solve consensus. 

\begin{definition}[Extended $k$-OSR PD]\label{def:eosr}
A knowledge connectivity graph $G_{\mathit{di}}=(V_{\mathit{di}},E_{\mathit{di}})$ belongs to extended $k$-OSR PD if it satisfies the following properties:
\begin{itemize}
    \item It belongs to $k$-OSR PD.
    \item There is a subgraph~$G_{\mathit{core}}=(V_{\mathit{core}},E_{\mathit{core}})$, namely the \textit{core}, that satisfies the following properties:
    \begin{enumerate}[label=C\arabic*), leftmargin=2em]
        %
        \item The core has the maximum connectivity among the sinks.
        Specifically, if any subset of processes other than the core members is considered as a sink, then its connectivity is less than the core's connectivity, i.e., 
        $\forall V \subseteq V_\mathit{di} : V \neq V_{\mathit{core}} \land \mathtt{isSink}_{G_{\mathit{di}}}^*(V) \Rightarrow \mathtt{k}_{G_{\mathit{di}}}(V_{\mathit{core}}) > \mathtt{k}_{G_{\mathit{di}}}(V)$.
        \item From any process~$i\notin V_{\mathit{core}}$ to any process~$j\in V_{\mathit{core}}$ there are at least $\mathtt{k}_{G_{\mathit{di}}}(V_{\mathit{core}})$ node-disjoint paths.
    \end{enumerate}   
\end{itemize}    
\end{definition}

We now provide an insight into the properties of the extended $k$-OSR PD: 
\begin{itemize}
    \item Property C1 ensures that there is only one core, i.e., multiple sinks cannot declare themselves as a core in a graph belonging to the extended $k$-OSR PD.
    Indeed, the core can be uniquely identified due to its maximum connectivity among the sinks.
    Furthermore, this property ensures that $\mathtt{k}_{G_{\mathit{di}}}(V_{\mathit{core}}) \geq k$. 
    Since a graph belonging to the extended $k$-OSR PD belongs also to the $k$-OSR PD, the graph has at least one sink with connectivity~$k$.
    As there is no sink with a connectivity greater than the core, $\mathtt{k}_{G_{\mathit{di}}}(V_{\mathit{core}}) \geq k$.
    \item Property C2 ensures that the core is inside the sink of~$G_{\mathit{di}}$.
    Note that if the core is outside of the sink of $G_{\mathit{di}}$, then the correct sink members of $G_{\mathit{di}}$ must have at least one path to the core due to this property, which is not possible.
    Additionally, this property ensures that non-core members can discover the core, like the non-sink members that discover the sink in a graph belonging to $k$-OSR PD.
    %
\end{itemize}

When a knowledge connectivity graph satisfies the following two properties, we say that it satisfies the requirements of the BFT-CUPFT model: 
\begin{enumerate*}[label=(\alph*)]
    \item its safe subgraph $G_{\mathit{safe}}$ belongs to the extended $(f+1)$-OSR PD, and 
    \item the core component of $G_{\mathit{safe}}$ contains at least $2f+1$ processes.
\end{enumerate*}

The graphs depicted in Fig.~\ref{fig:rich:k:osr:all} are examples of knowledge connectivity graphs that satisfy the requirements of the BFT-CUPFT.
The graph of Fig.~\ref{fig:rich:k:osr} shows an example of how processes in $\{5,6,7,8\}$ were prevented from mistakenly identifying themselves as a sink by adding two extra links from process~$6$ to process~$3$ and from process~$7$ to process~$2$.
In the resulting graph, processes in $\{5,6,7,8\}$ cannot identify themselves as a sink, as they discover the existence of processes~$1$,~$2$,~$3$, and~$4$ even if process~$5$ is slow.
Similarly, the graph of Fig.~\ref{fig:rich:k:osr:2} shows an example of how processes in $\{1,\dots, 7\}$ were prevented from mistakenly identifying themselves as a sink.

\section{Consensus in the BFT-CUPFT Model}\label{sec:consensus:extended:bft:cup}
This section presents a protocol that solves consensus in the BFT-CUPFT model. 
Recall that in Section~\ref{sec:revisiting:bft:cup}, we presented three algorithms~-- Discovery, Sink, and Consensus~-- to solve consensus in the authenticated BFT-CUP model. 
Since processes can execute the Discovery and Consensus algorithms without knowing the fault threshold, we can employ them in scenarios where the fault threshold is unknown.
Since there is only one core in a knowledge connectivity graph that satisfies the requirements of the BFT-CUPFT model, we design an algorithm called Core that allows processes to discover the core component.
We then use the Core algorithm instead of the Sink algorithm in the Consensus algorithm for solving consensus in the BFT-CUPFT model.

\vspace{0.2em}
\noindent\textbf{The Core algorithm in the BFT-CUPFT model.}
The objective of the Core algorithm in the BFT-CUPFT model is to allow each process to discover the core members.
In further detail, each process expands the set of processes it knows by executing the Discovery algorithm.
It executes the $\mathtt{discovery}$ task until it identifies the core component.
Afterward, any correct process $i$ terminates by returning a set that contains all core members.
The following theorem determines how a process can identify whether a set of processes is the core.

\begin{theorem}\label{thm:sink:alternative:def:wf}
    In a knowledge connectivity graph that satisfies the requirements of the BFT-CUPFT model, a subgraph with the node set~$V_{\mathit{core}}$ is the core if
    \begin{enumerate*}[label=(\alph*)]
        \item $\mathtt{isSink}_{G_{\mathit{di}}}^*(V_{\mathit{core}}) = \textit{true}$, and 
        \item there is no set $V\subset V_{\mathit{core}}$ such that $\mathtt{isSink}_{G_{\mathit{di}}}^*(V) = \textit{true}$ and $\mathtt{k}_{G_{\mathit{di}}}(V) \geq \mathtt{k}_{G_{\mathit{di}}}(V_{\mathit{core}})$.
    \end{enumerate*}
\end{theorem}

\ifbool{extendedVersion}{
The proof of the above theorem and the other theorems presented in this section can be found in Appendix~\ref{appendix:consensus:bft:cupft}.
}{}
The Core algorithm is shown in Algorithm~\ref{alg:sink}.
This algorithm provides only one function, $\mathtt{core}$.
Upon initiation, the algorithm forks the $\mathtt{discovery}$ task, allowing concurrent execution of the Discovery algorithm. 
Process $i$ waits until the properties specified in Theorem~\ref{thm:sink:alternative:def:wf} are satisfied for a subset of processes. 
Specifically, it waits until there exists two sets $S_1$ and $S_2$ with the conditions specified in line~\ref{line:condition:core}.
These conditions are analogous to the properties specified in Theorem~\ref{thm:sink:alternative:def:wf}.
The Core algorithm satisfies the properties specified in the following theorem.
\begin{theorem}\label{thm:sink:termination:wf}
    In a knowledge connectivity graph that satisfies the requirements of the BFT-CUPFT model, 
    Algorithm~\ref{alg:sink} executed by any correct process 
    \begin{enumerate*}[label=(\alph*)]
        \item eventually terminates and
        \item returns all core members.
    \end{enumerate*}
\end{theorem}
Finally, in order to solve consensus in the BFT-CUPFT model, we use the Core algorithm instead of the Sink algorithm in the Consensus algorithm outlined in Algorithm~\ref{alg:consensus}.
The following theorem states that if a knowledge connectivity graph satisfies the requirements of the BFT-CUPFT model, it is sufficient to allow processes to solve consensus.
\begin{theorem}\label{thm:last}
    In a knowledge connectivity graph that satisfies the requirements of the BFT-CUPFT model, 
    processes can solve consensus by executing the Core algorithm instead of the Sink algorithm in Algorithm~\ref{alg:consensus}.
\end{theorem}

\begin{algorithm}[t!]
\caption{The Core algorithm in the BFT-CUPFT model -- process $i$.}
\label{alg:sink}
\begin{algorithmic}[1]
\NoThen
\NoDo

\STATEx{\hspace{-1.67em}\textbf{function} $\mathtt{core}()$}
    \STATE{\textbf{fork} $\mathtt{discovery}()$}

    \STATE{\textbf{wait until} $\exists g\geq 0, \exists S_1\subseteq \mathcal{S}_\mathit{received}^i, \exists S_2 \subseteq \mathcal{S}_\mathit{known}^i \setminus S_1: \mathtt{isSink}_{G_{\mathit{di}}}(g, S_1, S_2) = \textit{true} \land
    (\forall g'>g, \forall Q_1 \subset S_1,$}
    \STATEx{$\quad \forall Q_2 \subseteq \mathcal{S}_\mathit{known}^i \setminus Q_1 : \mathtt{isSink}_{G_{\mathit{di}}}(g, Q_1, Q_2) = \textit{false})$}\label{line:condition:core}
    \STATE{\textbf{return} $S_1 \cup S_2$}

\end{algorithmic}
\end{algorithm}

\section{Related Work}\label{sec:related:work}
\noindent\textbf{Consensus with Unknown Participants.}
The evolution of the Consensus with Unknown Participants (CUP) problem has involved a series of advancements to accommodate diverse system models. 
Initially, Cavin et al.~\cite{cavin_2004} defined the problem for failure-free asynchronous systems, introducing a participant detector abstraction to provide initial information about system membership. 
The information collectively forms a knowledge connectivity graph, and that work establishes the necessary and sufficient properties that knowledge connectivity graphs must satisfy to solve the CUP problem. 
Subsequently, CUP was addressed in~\cite{cavin_2005} for crash-prone systems using the Perfect ($\mathcal{P}$) failure detector \cite{chandra_1996}. 
As implementing $\mathcal{P}$ requires synchrony, Greve and Tixeuil~\cite{greve_2007} relaxed the assumption to partial synchrony~\cite{dwork_1988} by augmenting the minimum required knowledge, specifically by increasing connections in the knowledge connectivity graph. 
This augmentation is demonstrated to be the minimum to tolerate crash failures without imposing synchrony requirements. 
The latest milestone extended CUP to tolerate Byzantine failures, introducing the BFT-CUP protocol~\cite{alchieri_2008, alchieri_2016}.

\vspace{0.2em}
\noindent\textbf{Consensus in directed graphs.}
Somewhat similar to the CUP model, several studies explore consensus in directed graphs, e.g.,~\cite{biely_2012,vaidya_2012,biely_2018,tseng_2015}.
Nevertheless, these investigations focus on determining the properties of the underlying communication graph to achieve consensus under diverse assumptions. 
For instance, Tseng and Vaidya~\cite{tseng_2015} established the minimal conditions of the underlying communication graph, where a participant $i$ can transmit messages to participant $j$ if a directed edge from $i$ to $j$ exists in the graph; otherwise, $i$ cannot send messages to $j$. 
Typically, these studies assume that the set of participants and the underlying communication graph are known to all participants. 
However, in the CUP model, the communication graph is complete, and the objective is to determine the necessary and sufficient initial knowledge about other participants required to solve consensus without knowing the system's membership.

\vspace{0.2em}
\noindent\textbf{Consensus on heterogeneous quorum systems.}
The exploration of protocols designed for systems where participants can have different trust assumptions, i.e., each participant can trust a subset of participants, originated in~\cite{damgard_2007}. 
Ripple~\cite{schawartz_2014,chase_2018} attempted to leverage this approach to address consensus in the permissionless setting, aiming to establish an efficient blockchain infrastructure. However, the achievement of this goal faced challenges leading to safety and liveness violations~\cite{amores_2020}. In contrast, Stellar~\cite{mazieres_2015,lokhava_2019}, based on the Federated Byzantine Quorum System (FBQS) formally studied later in~\cite{garcia_2018}, successfully achieved this objective.
In this approach, a network of trust emerges from the partial view declared by each participant. Consensus in this network is ensured if it adheres to the \emph{intact set} property, stipulating that all correct participants must form a quorum, and any two quorums formed by correct participants must intersect.

The connection between FBQS and dissemination Byzantine quorum systems~\cite{malkhi_1998} was established in~\cite{garcia_2018} and~\cite{garcia_2019}, showing the construction of a dissemination Byzantine quorum system corresponding to an FBQS. Subsequent work by Losa et al.~\cite{losa_2019} generalized FBQS to Personal Byzantine Quorum System (PBQS), demonstrating that consensus with weaker properties than the intact set is achievable through a \emph{consensus cluster}. Notably, forming a quorum by all correct participants is not mandatory within the consensus cluster.

Cachin and Tackmann~\cite{cachin_2019} extended Byzantine Quorum Systems (BQS)~\cite{malkhi_1998} from the symmetric trust model to the asymmetric model, facilitating a comparison between PBQS and the classical BQS model. Recently, Cachin et al. extended the asymmetric trust model, allowing each participant to make assumptions about the failures of participants it knows and, through transitivity, about failures of participants indirectly known by it~\cite{cachin_2022}.
Li et al.~\cite{li2023quorum} recently introduced heterogeneous quorum systems similar to PBQSs and demonstrated that the two properties~-- quorum intersection and availability~-- are necessary but insufficient to solve consensus.
They introduced the notion of quorum subsumption and established that the three conditions together are sufficient.
Furthermore, Vassantlal et al.~\cite{vassantlal_2023} showed that the Stellar consensus protocol cannot solve consensus when each participant has only the minimum knowledge required to solve consensus, even though BFT-CUP can.

This line of research diverges from the CUP model in the following aspects. 
While these studies presume each process possesses a local fault threshold, the CUP model operates under the assumption of a global fault threshold. 
Furthermore, these studies identify the properties of quorums in order to solve consensus, whereas the CUP model outlines the requirements for knowledge connectivity graphs.


\vspace{0.2em}
\noindent\textbf{Sleepy model.}
The CUP model addresses consensus in partially synchronous systems, accommodating correct or faulty participants. However, this model assumes that correct members remain actively engaged throughout the entire execution, which may be impractical in real-world scenarios.
In contrast, the sleepy model~\cite{momose_2022,pass_2017} introduces a different perspective. In this model, participants in a synchronous system are categorized as either awake or asleep, with awake participants capable of being either faulty or correct. The system's fault tolerance dynamically adjusts as participants transition between awake and asleep states. Crucially, consensus can be achieved if the majority of awake participants are correct at any given time. Moreover, unlike CUP, all participants in this model have knowledge of the system's membership.

\vspace{0.2em}
\noindent\textbf{Consensus using broadcast medium.}
Khanchandani and Wattenhofer~\cite{khanchandani_2021} established the impossibility of solving consensus in non-synchronous systems where participants and the fault threshold are unknown. In their system, \textit{no process initially possesses knowledge about other processes}, and each process utilizes a broadcast medium for communication. Recall that our primary goal in this paper is to solve consensus in partially synchronous systems, where each process lacks information about the system's membership and the fault threshold.
Accordingly, at first glance, our goal might appear contradictory to that impossibility result.
However, each process has knowledge about the existence of a subset of processes in our model (BFT-CUPFT), which is sufficient to discover the core component.
Recall that all processes discover the same core. 
The correct processes within the core can solve consensus and inform other processes about the decided value, resulting in solving consensus by all correct processes.
Hence, that impossibility result does not apply to our work.

\section{Conclusion}\label{sec:conclusion}
We addressed the critical challenge of solving Byzantine fault-tolerant consensus in partially synchronous systems where each participant joins the network by having only partial knowledge about the existence of other participants and without explicit information about the fault threshold. 
We demonstrated that the key challenge arises from the possibility of having multiple disjoint subsets of processes, each solving a distinct instance of consensus, thereby violating the Agreement property of consensus. 
In order to mitigate this issue, we specified the sufficient knowledge connectivity requirements that must be satisfied to allow solving consensus in such settings.

\section*{Acknowledgments}
This work was supported by FCT through 
the Ph.D. scholarship, ref. \href{https://doi.org/10.54499/2020.04412.BD}{2020.04412.BD}, 
the SMaRtChain project, ref. \href{https://doi.org/10.54499/2022.08431.PTDC}{2022.08431.PTDC}, and 
the LASIGE Research Unit, ref. \href{https://doi.org/10.54499/UIDB/00408/2020}{UIDB/00408/2020} and ref. \href{https://doi.org/10.54499/UIDP/00408/2020}{UIDP/00408/2020}.

\ifbool{extendedVersion}{}{}
\bibliographystyle{IEEEtran}
\bibliography{ref}

\section{Appendix}

\subsection{Correctness Proofs for the Discovery Algorithm in the Authenticated BFT-CUP Model}\label{appendix:discovery}
\begin{proof}[Proof of Theorem~\ref{thm:discovery:sink}]
Consider a system with a knowledge connectivity graph $G_{\mathit{di}} \in \mathcal{G}_\mathit{di}$. 
Suppose $V_{\mathit{sink}}$ comprises the sink members of $G_{\mathit{di}}$.
By executing Algorithm~\ref{alg:discovery:known:f}, we need to show that each correct process $i$ eventually
    \begin{enumerate*}[label=(\alph*)]
        \item discovers all correct sink members, i.e., 
        $V_{\mathit{sink}}\cap \Pi_C \subseteq \mathcal{S}_\mathit{known}^i$, and
        \item receives the PDs of all correct sink members, i.e.,
        $V_{\mathit{sink}} \cap \Pi_C \subseteq \mathcal{S}_\mathit{received}^i$.
    \end{enumerate*}
We divide the proof into the following two cases:
\begin{itemize}
    \item Process~$i$ is a correct sink member.
    Let $d_{\mathit{ss}}$ denote the longest distance between any two correct sink members.
    After GST, $i$ sends a \textsc{GetPDs} message to its PD members, and each correct process within $\mathit{PD}_i$ responds by sending its PD to $i$.
    Hence, process $i$ receives the PDs of correct sink processes within a distance of two from itself at most by time $\text{GST} + 2\delta$.
    It then sends a \textsc{GetPDs} message to the processes within a distance of two from itself, and each correct process that receives such a message responds by sending its PD to $i$.
    Consequently, it receives the PDs of correct sink processes within a distance of three from itself at most by time $\text{GST} + 4\delta$.
    Using the same argument, process~$i$ receives the PDs of all correct sink members at most by time $\text{GST} + 2(d_{\mathit{ss}}-1)\delta$, i.e., $V_{\mathit{sink}} \cap \Pi_C \subseteq \mathcal{S}_\mathit{received}^i$.
    Consequently, $ V_{\mathit{sink}} \cap \Pi_C \subseteq \mathcal{S}_\mathit{known}^i$.
    
    %
    \item Process~$i$ is a correct non-sink member.
    Let $d_{\mathit{ns}}$ denote the longest distance between a correct non-sink member and a correct sink member.
    After GST, process~$i$ sends a \textsc{GetPDs} message to its PD members, and each correct process within $\mathit{PD}_i$ responds by sending its PD to $i$.
    Hence, process~$i$ receives the PDs of correct processes within a distance of two from itself at most by time $\text{GST} + 2\delta$.
    It then sends a \textsc{GetPDs} message to the processes within a distance of two from itself, and each correct process that receives such a message responds by sending its PD to $i$.
    Consequently, it receives the PDs of correct processes within a distance of three from itself at most by time $\text{GST} + 4\delta$.  
    Using the same argument, process~$i$ receives the PD of at least one correct sink member~$j$ at most by time $\text{GST} + 2(d_{\mathit{ns}}-1)\delta$, as there are at least $f+1$ node-disjoint paths made by correct processes from $i$ to $j$.
    Hence, $i$ discovers all correct sink members and receives their PDs at most by time $\text{GST} + 2(d_{\mathit{ns}} + d_{\mathit{ss}} - 2)\delta$, i.e., $ V_{\mathit{sink}} \cap \Pi_C \subseteq \mathcal{S}_\mathit{known}^i$, and
    $V_{\mathit{sink}} \cap \Pi_C \subseteq \mathcal{S}_\mathit{received}^i$.
\end{itemize}
\end{proof}

\subsection{Correctness Proofs for the Sink Algorithm in the Authenticated BFT-CUP Model}\label{appendix:sink}
\begin{proof}[Proof of Theorem~\ref{thm:sink:alternative:def}]
    Let $G_{\mathit{di}}=(V_{\mathit{di}},E_{\mathit{di}})$ be a knowledge connectivity graph belonging to $\mathcal{G}_{\mathit{di}}$.
    Assume that $V_{\mathit{sink}}$ contains all sink members of this graph.
    We show that $V_{\mathit{sink}}$ can be partitioned into two sets $S_1,S_2 \subseteq V_{\mathit{di}}$ such that these sets satisfy Properties P1, P2, P3, and P4 for a given fault threshold~$f$.
    To do so, assume set~$S_1$ contains all and only correct sink members and set~$S_2$ contains all and only Byzantine sink members.
    \begin{itemize}
        \item Due to Theorem~\ref{thm:bft:cup}, there are at least $2f+1$ correct processes inside the sink.
        Hence $|S_1|\geq 2f+1$, which means Property P1 is satisfied.
        \item Due to Theorem~\ref{thm:bft:cup}, the strong connectivity of the correct sink members is at least $f+1$.
        Therefore, $\kappa(S_1)\geq f+1$, which means that Property P2 is satisfied.
        \item Note that no correct sink member has a link to non-sink members.
        Besides, the number of Byzantine sink members is bounded by~$f$.
        Consequently, the number of processes in~$S_1$ that have links to~$V_{\mathit{di}}\setminus S_1$ is bounded by~$f$, satisfying Property P3.
        \item Since $S_2$ is a part of the sink, for each process~$i \in S_2$, there are at least $f+1$ members of $S_1$ that have outgoing links to~$i$, satisfying Property P4.
    \end{itemize}

\end{proof}

\begin{proof}[Proof of Theorem~\ref{thm:sink:exists}]
    We need to show that set~$S_1 \cup S_2$ contains all and only the sink members.
    We divide the proof into the following two steps:
    \begin{itemize}
        \item Set~$S_1 \cup S_2$ contains only the sink members.
            We also divide this step into the following two sub-steps:
            \begin{itemize}
                \item $S_1$ contains only the sink members.
                    For the sake of contradiction, assume that $S_1$ does not contain only the sink members.
                    Hence, there is at least a non-sink member~$i \in S_1$.
                    From properties P1 and P2, $|S_1| \geq 2f+1$ and $\kappa(G_{\mathit{di}}[S_1]) \geq f+1$.
                    As there are at most~$f$ Byzantine processes in the system, $S_1$ has at least $f+1$ correct processes.
                    Since correct sink members do not have any link to non-sink members, the correct members of $S_1$ do not have any link to $i$.
                    Hence, $\kappa(G_{\mathit{di}}[S_1]) \not\geq f+1$, which is a contradiction.
                \item $S_2$ contains only the sink members.
                    For the sake of contradiction, assume that $S_2$ does not contain only the sink members.
                    Hence, there is at least a non-sink member~$i \in S_2$.
                    Due to the previous sub-step, all members of $S_1$ are sink members.
                    From property P4, there should be at least $f+1$ processes in~$S_1$ that have outgoing links to~$i$.
                    Since there are at most $f$ Byzantine processes in~$S_1$, at least a correct member of $S_1$ has an outgoing link to~$i$, which is impossible as no correct sink member has a link to a non-sink member.
            \end{itemize}
        \item Set~$S_1 \cup S_2$ contains all sink members.
        For the sake of contradiction, assume that $S_1 \cup S_2$ does not contain all sink members.
        Hence, there is at least a sink member~$i \notin S_1 \cup S_2$.
        From the previous step, $S_1$ contains only the sink members.
        Due to property P1, $|S_1| \geq 2f+1$.
        As there are at most~$f$ Byzantine processes in the system, $S_1$ has at least $f+1$ correct processes.
        Note that $i$ can be either correct or Byzantine.
        \begin{itemize}
            \item If $i$ is correct, as each correct sink member has at least $f+1$ node-disjoint paths made by correct sink members to other correct sink members, each correct member of~$S_1$ has at least $f+1$ node-disjoint paths to~$i$.
            Hence, at least $f+1$ correct members of $S_1$ have links to the outside of $S_1$, violating property P3.
            Accordingly, set $S1 \cup S_2$ contains all correct sink members.
            \item If $i$ is Byzantine and a sink member, there are at least~$f+1$ correct sink members that have outgoing links to~$i$.
            Hence, $i \in S_2$ due to property P4, which is a contradiction.
        \end{itemize}
    \end{itemize}
    From these two steps, it follows that set~$S_1 \cup S_2$ contains all and only the sink members.
\end{proof}

\begin{proof}[Proof of Theorem~\ref{thm:sink:termination}]
    In a knowledge connectivity graph $G_{\mathit{di}} \in \mathcal{G}_{\mathit{di}}$, we need to show that Algorithm~\ref{alg:sink:known:f} executed by any correct process has the following properties:
    \begin{enumerate}[label=(\alph*)]
        \item It eventually terminates.
        From Theorem~\ref{thm:discovery:sink}, any correct process discovers all correct sink members at most by time $t = \mathtt{max}\{\text{GST} + 2(d_{\mathit{ss}} - 1)\delta,\text{GST} + 2(d_{\mathit{ns}} + d_{\mathit{ss}} - 2)\delta\}$.
        We show that the conditions presented in line~\ref{line:sink:condition} within Algorithm~\ref{alg:sink:known:f} are met at time $t$.
        As a consequence, any correct process stops to await further, resulting in the termination of the algorithm.
        
        Since any correct process discovers all sink members at time~$t$, we have $\forall i \in V_{\mathit{sink}} \cap \Pi_C : V_{\mathit{sink}}  \subseteq \mathcal{S}_\mathit{known}^i$.
        Besides, at time~$t$, we have $\forall i \in V_{\mathit{sink}} \cap \Pi_C : V_{\mathit{sink}} \cap \Pi_C \subseteq \mathcal{S}_\mathit{received}^i$.
        Since there are at most $f$ Byzantine sink members that can remain silent, the number of processes that each process $i \in V_{\mathit{sink}} \cap \Pi_C$ knows, but it has not received its PD is limited by $f$.
        Consequently, the two sets $S_1$ and $S_2$ specified in the conditions exist, allowing each process~$i$ to be able to terminate.
        \item It returns all sink members.
        This part directly follows from Theorem~\ref{thm:sink:alternative:def}.
        Specifically, sets $S_1$ and $S_2$ in Algorithm~\ref{alg:sink} are those defined in Theorem~\ref{thm:sink:alternative:def}.
    \end{enumerate}
\end{proof}

\subsection{Correctness Proofs for the Consensus Algorithm in the Authenticated BFT-CUP Model}\label{appendix:consensus}
\begin{proof}[Proof of Theorem~\ref{thm:consensus}]
    From Theorem~\ref{thm:sink:termination}, any correct process discovers the sink eventually.
    Upon discovering the sink, sink members execute a consensus protocol.
    Since there are at least $2f+1$ correct sink members, while there are at most $f$ Byzantine sink members, sink members can solve consensus using a traditional BFT consensus protocol like PBFT~\cite{castro_1999}.
    Recall that non-sink members request the decided value from sink members.
    When a correct sink member decides on a value, it sends the value to non-sink members.
    A non-sink member $i$ decides on a value after receiving a value from $f+1$ sink members.
    This way, process $i$ can ensure that there is at least one correct process among the processes sent the decided value, as there are at most $f$ Byzantine sink members.
    We now show that the properties of consensus specified in Subsection~\ref{subsection:consensus} are satisfied when processes execute Algorithm~\ref{alg:consensus}.
    \begin{itemize}
        \item Validity.
        Since a traditional BFT consensus protocol like PBFT~\cite{castro_1999} satisfies the Validity property, the decided value by correct sink members satisfies this property.
        As any correct non-sink member $i$ for deciding a value requires receiving the decided value by at least a correct sink member, the decided value by $i$ also satisfies the Validity property.
        Hence, the value decided by all correct processes satisfies the validity property.
        \item Agreement.
        Since a traditional BFT consensus protocol like PBFT~\cite{castro_1999} satisfies the Agreement property, the decided value by correct sink members satisfies this property.
        As any correct non-sink member $i$ for deciding a value requires receiving the decided value by at least a correct sink member, the decided value by $i$ must be equal to the sink members, satisfying the Agreement property.
        \item Termination.
        Since a traditional BFT consensus protocol like PBFT~\cite{castro_1999} satisfies the Termination property.
        Hence, any correct sink member will eventually terminate.
        Since there are at least $2f+1$ correct sink members, and any non-sink member $i$ is required to receive $f+1$ messages to terminate, process $i$ eventually terminates. 
    \end{itemize}
    Since all properties of consensus are satisfied, Algorithm~\ref{alg:consensus} solves consensus.
\end{proof}

\subsection{Correctness Proofs Related to Solving Consensus in the BFT-CUPFT Model}\label{appendix:consensus:bft:cupft}
\begin{proof}[Proof of Theorem~\ref{thm:sink:alternative:def:wf}]
    This theorem directly follows from Properties C1 and C2.
\end{proof}

\begin{proof}[Proof of Theorem~\ref{thm:sink:termination:wf}]
    The proof of this theorem is similar to the proof of Theorem~\ref{thm:sink:termination}.
\end{proof}


\begin{proof}[Proof of Theorem~\ref{thm:last}]
    Recall that in Theorem~\ref{thm:consensus}, correct processes can solve consensus due to the existence of a sink in the knowledge connectivity graph.
    As there is only one core in a knowledge connectivity graph that belongs to the BFT-CUPFT model, the proof of this theorem is similar to the proof of Theorem~\ref{thm:consensus}.
\end{proof}

\end{document}